\documentclass[letterpaper, 10 pt, conference]{ieeeconf}

\IEEEoverridecommandlockouts                              
\usepackage{ifpdf}

\ifpdf
   \usepackage[pdftex]{graphicx}
   \graphicspath{{./pic/}}
   \DeclareGraphicsExtensions{.pdf,.jpeg,.jpg,.png}
\else
   \usepackage[dvips]{graphicx}
   \graphicspath{{./eps/}}
   \DeclareGraphicsExtensions{.eps}
\fi

\usepackage[caption=false,font=footnotesize]{subfig}

\usepackage[cmex10]{amsmath}
\usepackage{array}
\usepackage{amssymb}
\newtheorem{theorem}{Theorem}

\usepackage{algorithmic}
\begin{document}
%
\title{\LARGE \bf Distributed Successive Approximation Coding\\using Broadcast Advantage: The Two-Encoder Case}

\author{
\authorblockN{Zichong Chen, Guillermo Barrenetxea, and Martin Vetterli}
\authorblockA{LCAV - School of Computer and Communication Sciences\\
Ecole Polytechnique F\'{e}d\'{e}rale de Lausanne (EPFL), Lausanne CH-1015, Switzerland\\
Email: \{zichong.chen, guillermo.barrenetxea, martin.vetterli\}@epfl.ch
}
\thanks{This research was supported by the National Competence Center in Research on Mobile Information
and Communication Systems (NCCR-MICS, http://www.mics.org), a center supported by the Swiss
National Science Foundation.}
}


%


\maketitle

\begin{abstract}
Traditional distributed source coding rarely considers the possible link between separate encoders. However, the broadcast nature of wireless communication in sensor networks provides a free gossip mechanism which can be used to simplify encoding/decoding and reduce transmission power. Using this broadcast advantage, we present a new two-encoder scheme which imitates the ping-pong game and has a successive approximation structure. For the quadratic Gaussian case, we prove that this scheme is successively refinable on the \textit{\{sum-rate, distortion pair\}} surface, which is characterized by the rate-distortion region of the distributed two-encoder source coding. A potential energy saving over conventional distributed coding is also illustrated. This ping-pong distributed coding idea can be extended to the multiple encoder case and provides the theoretical foundation for a new class of distributed image coding method in wireless scenarios.
\end{abstract}


%

\section{Introduction}
\label{sec1}
The availability of low-cost image sensor chips such as CMOS cameras is shifting the paradigm of sensor network communication from regular and small data transmissions to occasional and large amounts of data communications. To deliver images under the severe energy constraints of wireless sensor networks (e.g., nodes operating on batteries and a solar panel), multiterminal source coding has an important role as it can cut the rate to the theoretical lower limit. One possible topology for a multi-camera network is to avoid multi-hop routing, which can be well modeled as a typical \textit{distributed source coding} setup where $N$ separate cameras transmit correlated images to a common base station (BS).

Looking at the typical multi-camera network setup as in Fig.~\ref{broadcast}, the cameras within the transmission range of the emitting camera overhear the transmitted messages. We can take advantage of this broadcast nature, as it is given for free and can potentially help to simplify the coding procedure and reduce the transmission power. It will be shown later that this broadcast advantage can be used in a successive manner between separate encoders, which also avoids interference (only one node transmits at the same time). Furthermore, since the scale of the local camera network is usually small compared with their distance to the BS, it is practical to assume the channel capacity between neighboring cameras is larger than the capacity between the camera and the BS. Therefore, for simplicity, we assume in this paper that all the overhearing nodes within the transmission range can get the message error-free, as long as it is received at the BS.

\begin{figure}[!ht]
\centering
\subfloat[]{\label{broadcast}\includegraphics[width=2.5in]{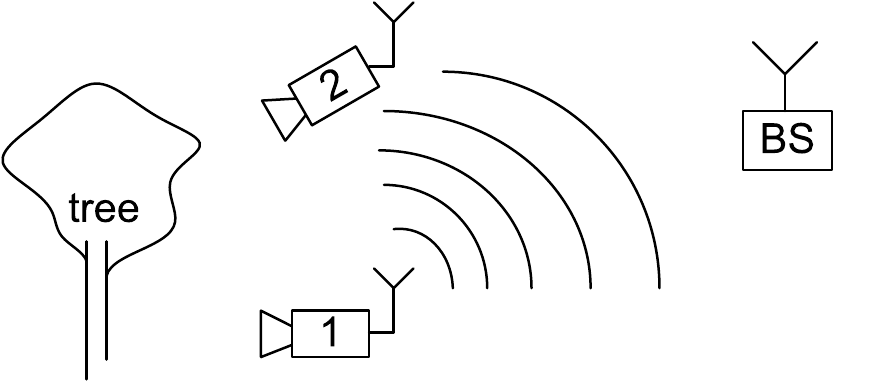}}
\\ \subfloat[]{\label{model}\includegraphics[width=2.5in]{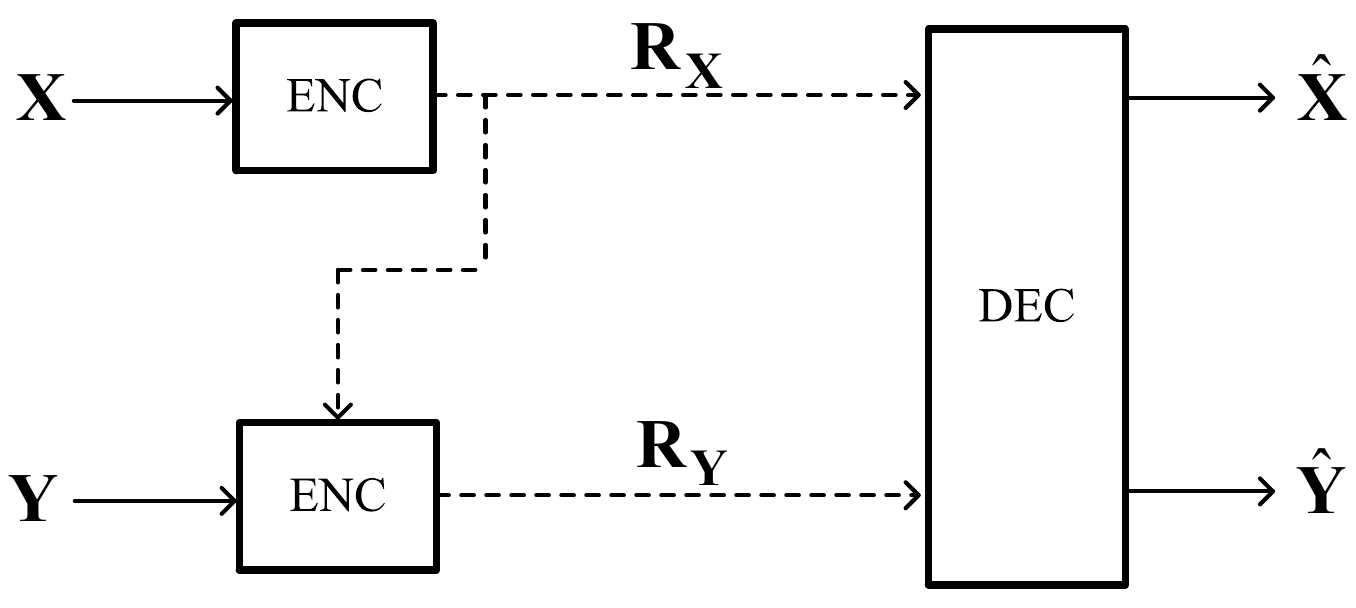}}
\caption{The typical multi-camera network: (a) When Node-1 transmits a message to the base station, Node-2 can overhear it. (b) Setup of the two-encoder distributed source coding with broadcast advantage.}
\end{figure}

To study the rate-distortion region of distributed source coding with broadcast advantage, we start with the simplest configuration of two encoders and one BS that are all within the transmission range of each other (see Fig.~\ref{model}): the source $X$ is encoded at rate $R_X$ without access to source $Y$. Based on the coded version of $X$, source $Y$ is encoded at the rate $R_Y$, which can be meanwhile regarded as a successive description for $X$. An extreme case for this scheme is when $R_X\geq H(X)$, then it reduces to the rate-distortion problem of $Y$ with side information $X$ fully available at encoder and decoder.

Successive approximation coding is another important capability desirable for a camera network. It has been primarily employed in image coding, where one gets better image quality step by step using successive descriptions. This is particularly useful in an energy-limited communication setup, since we can decide at the receiver whether a high resolution image is really needed after the low resolution version is displayed. Inspired by the ping-pong game, we extend Fig.~\ref{model} to a new distributed source coding scheme, which has a successive approximation structure. As Fig.~\ref{extend} shows, the broadcast messages act like a ping-pong ball, which is flipped back and fourth between the two encoders. We call such a scheme \textit{two-encoder Distributed Successive Approximation Coding using Broadcast Advantage (DiSAC2)}.

\begin{figure}[!ht]
\centering
\includegraphics[width=2.6in]{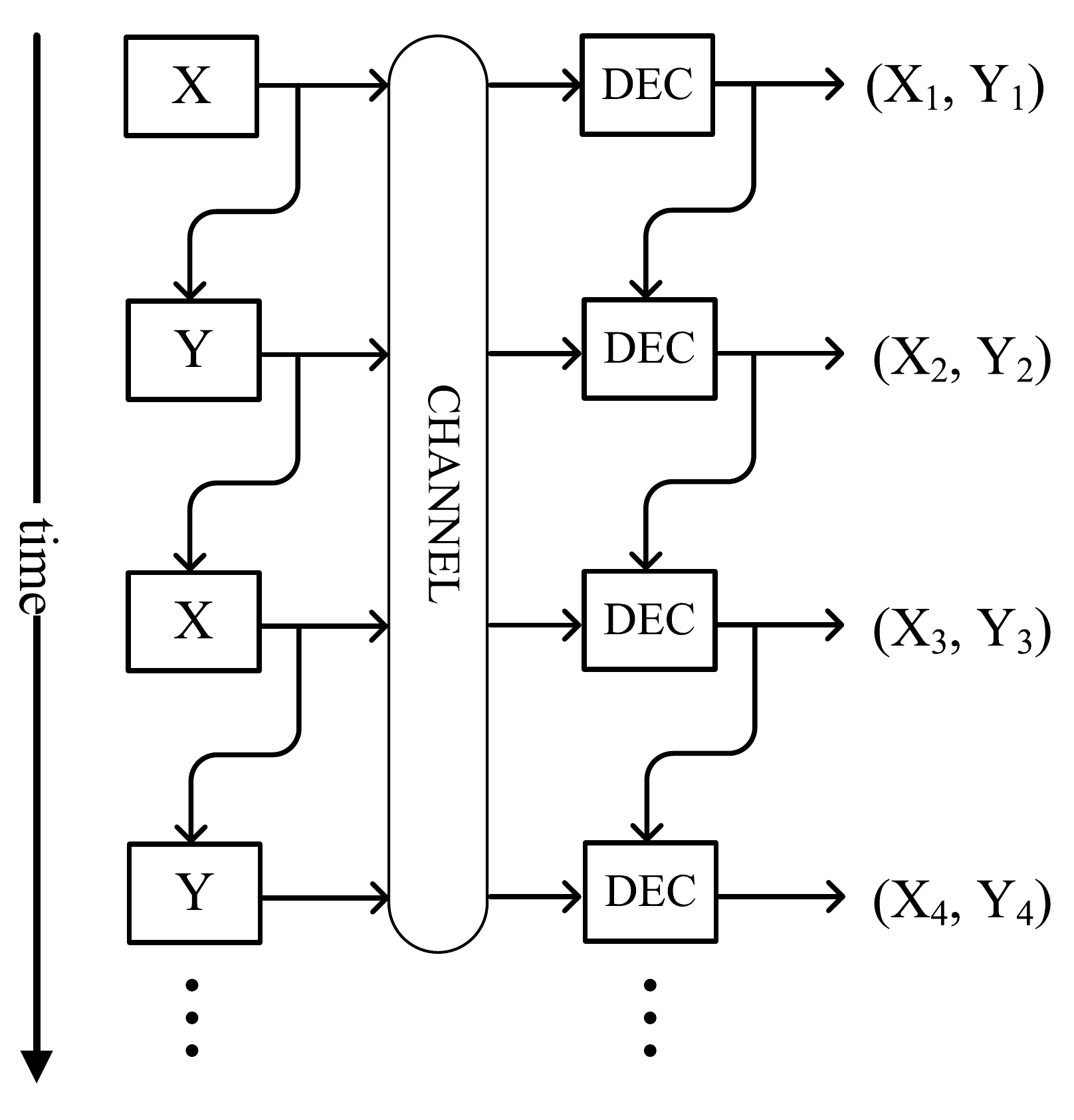}
\caption{Successive approximation structure of the two-encoder distributed source coding using broadcast advantage.}
\label{extend}
\end{figure}

In this paper, we present this new scheme and show that for the quadratic Gaussian case, it is successively refinable on the \textit{\{sum-rate, distortion pair\}} surface, which is characterized by the rate-distortion region of the \textit{distributed two-encoder source coding (DSC2)}~\cite{Wagner2008}. Furthermore, we show that it has the possibility for rate allocation and the potential for energy saving. The extension of the DiSAC2 to three and more encoders is also briefly addressed.

The remainder of this paper is organized as follows: In Sec.~\ref{sec2}, we review the related works about multiterminal source coding and successive approximation coding. In Sec.~\ref{sec3}, we present the precise setup and detailed analysis for a particular case of the DiSAC2 \textit{(three-stage, with Gaussian sources and quadratic distortion)}. Sec.~\ref{sec4} provides an extension of the new scheme to three and more encoders.

\section{Related work}
\label{sec2}
\subsection{Multiterminal source coding}
\label{sec2.1}
The general multiterminal source coding problem~\cite{Berger1978} has been open for thirty years. The rate region for the distributed lossless source coding problem has been solved by Slepian and Wolf~\cite{Slepian1973}. However, the general lossy case is not fully determined yet. Wyner and Ziv~\cite{Wyner1976} solved a special case when one of two sources is entirely known at the decoder. Recently, Wagner et al.~\cite{Wagner2008} gave the rate-distortion region for the two-encoder quadratic Gaussian case. Practical distributed coding schemes~\cite{Zamir2002}~\cite{Girod2005}~\cite{Puri2007} have been developed in recent years, however mostly restricted to sources with strong correlation such as 1-D measurement data, video frames, etc..

Distributed source coding using broadcast advantage is a particular form of the source coding problem with partially separated encoders. The rate-distortion region for two-encoder case was first addressed in \cite{Kaspi1982}, where coding theorems are determined for two cases: (i) one source is reproduced perfectly at the receiver; (ii) one source is perfectly revealed to the other source. \cite{Korner1983} and \cite{Oohama1996} further develop this idea to the case when the encoder in a lossless Slepian-Wolf setup can observe the coded data from the other encoder. It is proved that the admissible rate region is not enlarged. Fig.~\ref{model} is actually the lossy case for this setup, however the answer for its rate-distortion region is still unknown today.

\subsection{Successive approximation coding}
\label{sec2.2}

The optimality of the successive approximation scheme for a single source has been studied by Equitz and Cover~\cite{Equitz1991}: To encode a source $X$, a coarse description $\hat{X}_1$ with R-D pair $(R_1,D_1)$ is refined to a finer description $\hat{X}_2$ with R-D pair $(R_2,D_2)$. This scheme is called successively refinable when any rate pair $(R_1,R_2)$ operates on the rate-distortion function $R(D)$: $R_1=R(D_1)$ and $R_1+R_2=R(D_2)$. It is shown that this optimality is achieved if and only if we can write $\hat{X}_1\to \hat{X}_2 \to X$ as a Markov chain. A Gaussian source with quadratic distortion is one example that is successively refinable on the $R(D)$ curve.

There is related work on successive coding for multiple sources. \cite{Viswanathan2000} proposes the sequential coding of correlated sources for video applications, in which the first source is encoded solely while the subsequent source is encoded based on both sources. This scheme is a weak version of centralized coding as it has access to both sources. However, it does not fully exploit the joint information due to the first step encoding, therefore the minimum sum-rate is sometimes worse than DSC2. Recently, \cite{Behroozi2009} proposed a successive decoding scheme for the distributed source coding problem (no link between encoders). It is proved that successive decoding following a linear fusion could achieve the rate-distortion region of DSC2 for the quadratic Gaussian case. However, the final step of fusion actually breaks the successive decoding structure: all results have to be reconstructed after everything is received.

\section{Three-stage DiSAC2}
\label{sec3}
\subsection{Setup and notations}
\label{sec3.1}
\begin{figure*}[!ht]
\centerline{\subfloat[]{\label{DiSAC2}\includegraphics[width=2.6in]{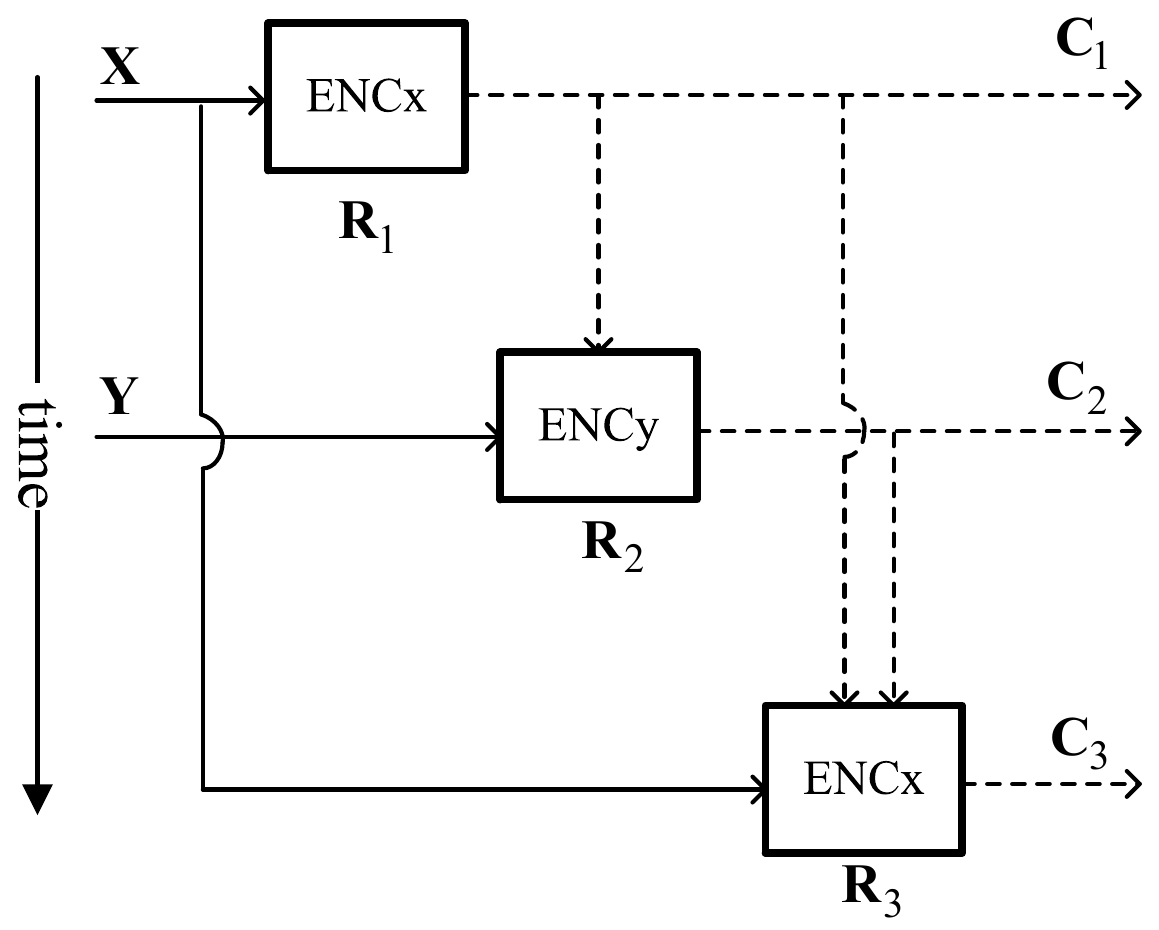}}
\hfil \subfloat[]{\label{multiSR}\includegraphics[width=2.6in]{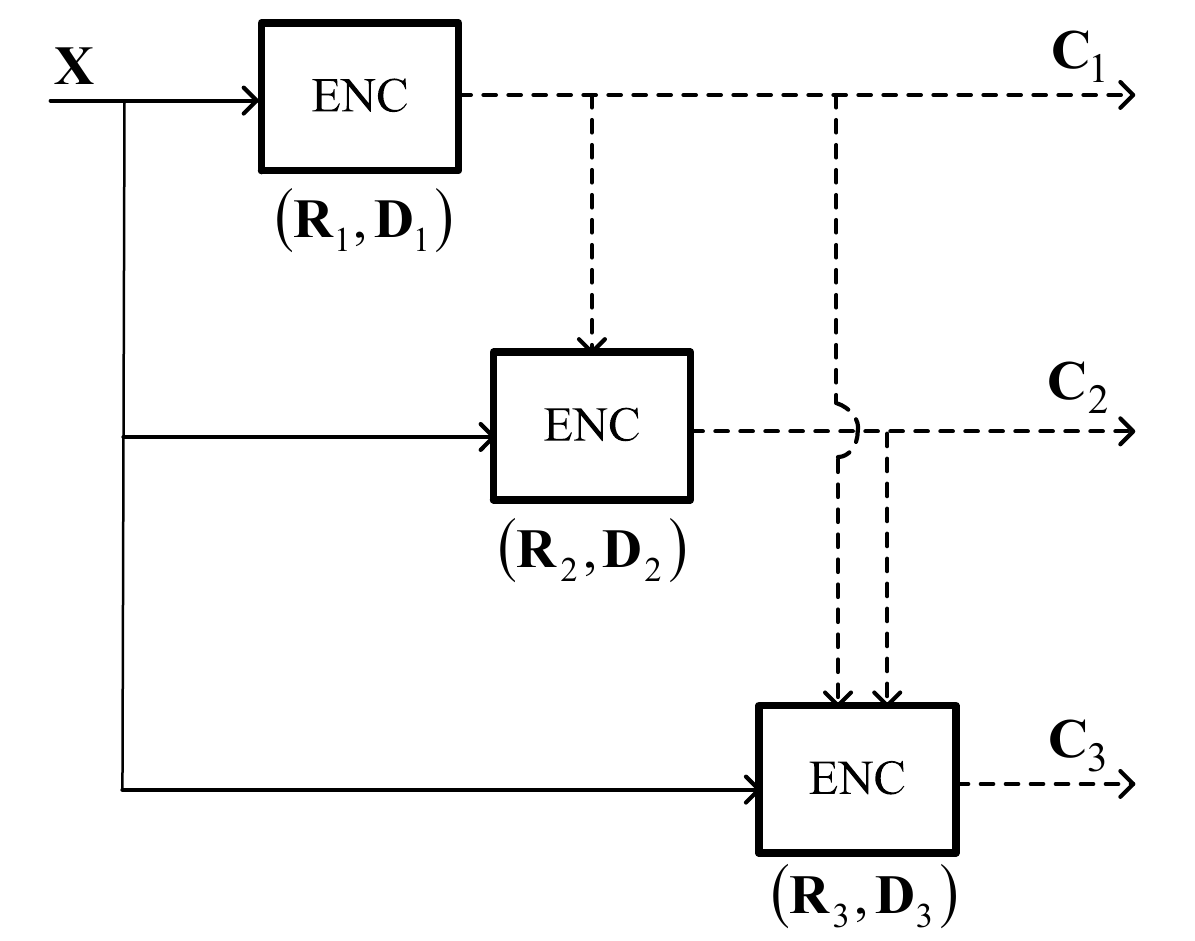}}}
\caption{Successive approximation coding schemes --- encoding part: (a) Three-stage DiSAC2 for two sources. (b) Three-stage successive refinement for a single source.}
\end{figure*}

For the DiSAC2 scheme introduced in Fig.~\ref{extend}, we can specify the number of encoding/decoding stages. Fig.~\ref{DiSAC2} shows the encoding part of a three-stage setup: two separate encoders ENCx and ENCy cooperate using broadcast advantage to convey the correlated sources $(X,Y)$ to the BS. The sketch for the coding procedure is as follows:
\begin{enumerate}
\item At the first stage, ENCx encodes $X$ to the codeword $C_1$ without any knowledge of $Y$, and has a rate of $R_1$. $(X_1, Y_1)$ is reconstructed at the BS after $C_1$ is received.
\item At the second stage, ENCy overhears the coded data $C_1$ which is being transmitted to the BS, and it only transmits the refinement which fully exploits the joint information between the source $Y$ and the coded data $C_1$. $C_2$ is the corresponding codeword sent in the second stage, which has a rate of $R_2$. $(X_2, Y_2)$ is reconstructed at the BS based on $(C_1,C_2)$.
\item Similarly, at the third stage, ENCx encodes $X$ to the codeword $C_3$ based on the coded data $C_1$ and $C_2$, which has a rate of $R_3$. $(X_3, Y_3)$ is reconstructed at the BS based on $(C_1,C_2,C_3)$.
\end{enumerate}

If $X=Y$ in Fig.~\ref{DiSAC2}, then it reduces to a three-stage successive approximation coding of a single source (see Fig.~\ref{multiSR}). As we know from \cite{Equitz1991}, it is successively refinable on the \textit{\{rate, distortion\}} curve, as long as the Markov condition is satisfied. Similar results can be investigated for the DiSAC2 scheme. To give an insight, we specifically discuss the three-stage DiSAC2 (as depicted in Fig.~\ref{DiSAC2}) with Gaussian sources and quadratic distortion:

Let $\mathcal{N}(\boldsymbol{\mu},\Sigma)$ be the notation for a multivariate Gaussian distribution with mean vector $\boldsymbol{\mu}$ and covariance matrix $\Sigma$. $(X,Y)\sim\mathcal{N}(\boldsymbol{\mu},\Sigma)$, where
\begin{equation*}
\boldsymbol{\mu}=(0,0),\quad \Sigma=\left(\begin{array}{cc}
                            1 & \rho \\
                            \rho & 1 \\
                            \end{array}
                        \right)\ \textrm{for }|\rho|<1.
\label{eq3.1.1}
\end{equation*}

The corresponding distortion pairs for the three stages are $(D_{X_1},D_{Y_1})$, $(D_{X_2},D_{Y_2})$, $(D_{X_3},D_{Y_3})$, respectively. The distortions are defined as $\mathbb{E} d(X,\hat{X})$ where $d(\cdot,\cdot)$ is the quadratic error measure, and similarly for $Y$.

\subsection{Coding procedure}
\label{sec3.2}
\subsubsection{First stage}
\label{sec3.2.1}
$X\sim \mathcal{N}(0,1)$ is encoded to $C_1$ using the random codebook argument, with R-D pair
\begin{equation}
R_1=\frac{1}{2}\log \frac{1}{D_{X_1}},\quad D_{X_1}\leq 1.
\label{eq3.2.1.1}
\end{equation}

$X_1$ can be decoded as: $X_1=C_1$, thus
\[\mathbb{E}d(X,X_1)=D_{X_1}.\]

Since $p(X)= \mathcal{N}(0,1)$, the \textit{test channel} in the first stage is
\begin{equation*}
p(X|C_1)=\mathcal{N}(C_1,D_{X_1}).
\end{equation*}
Substituting $X_1=C_1$,
\begin{equation}
p(X|C_1)=\mathcal{N}(X_1,D_{X_1}).
\label{eq3.2.1.2}
\end{equation}

To decode $Y_1$, we calculate the conditional probability
{\allowdisplaybreaks
\begin{align}
p(Y|C_1) &=\int\limits_{-\infty}^{+\infty} p(YX|C_1)\ \textrm{d}X \nonumber \\
 &=\int\limits_{-\infty}^{+\infty} p(Y|C_1X)\cdot p(X|C_1)\ \textrm{d}X \nonumber \\
 &=\int\limits_{-\infty}^{+\infty} p(Y|X)\cdot p(X|C_1)\ \textrm{d}X,
\label{eq3.2.1.3}
\end{align}}
where the third equality follows from the fact that $C_1$ is encoded and decoded from $X$ (a definite function of $X$), thus $p(Y|C_1X)=p(Y|X)$. From the joint distribution of $(X,Y)$,
\begin{equation}
p(Y|X)= \mathcal{N}(\rho X,1-\rho^2).
\label{eq3.2.1.3.1}
\end{equation}
Substituting (\ref{eq3.2.1.3.1}) and (\ref{eq3.2.1.2}) into (\ref{eq3.2.1.3}) leads to
\begin{equation}
p(Y|C_1)= \mathcal{N}(\mu_1,\sigma^2_1),
\label{eq3.2.1.4}
\end{equation}
where
\begin{equation}
\mu_1=\rho X_1,
\label{eq3.2.1.5}
\end{equation}
and
\begin{equation}
\sigma^2_1=D_{X_1} \rho ^2-\rho ^2+1.
\label{eq3.2.1.6}
\end{equation}

$Y_1$ can be decoded as: $Y_1=\mu_1$, thus
\[\mathbb{E}d(Y,Y_1)=\sigma^2_1.\]

To sum up, the distortion pair in the first stage is
\begin{equation}
(D_{X_1},D_{Y_1})=\left(D_{X_1},\sigma^2_1\right).
\label{eq3.2.1.7}
\end{equation}

\subsubsection{Second stage}
\label{sec3.2.2}
$C_1$ is known due to the broadcast advantage, thus according to (\ref{eq3.2.1.4}), $Y-\mu_1\sim \mathcal{N}(0,\sigma^2_1)$. $Y-\mu_1$ is encoded to $C_2$ using the random codebook argument, with R-D pair
\begin{equation}
R_2=\frac{1}{2}\log \frac{\sigma^2_1}{D_{Y_2}},\quad D_{Y_2}\leq \sigma^2_1.
\label{eq3.2.2.1}
\end{equation}

$Y_2$ can be decoded as: $Y_2=C_2+\mu_1$, thus
\begin{align*}
\mathbb{E}d(Y,Y_2)&=\mathbb{E}d(Y-\mu_1,Y_2-\mu_1) \\
&=\mathbb{E}d(Y-\mu_1,C_2) \\
&=D_{Y_2}.
\end{align*}

Similar calculations as Sec.~\ref{sec3.2.1}\footnote{Calculations are omitted due to space, see \cite{Chen2010b} for details.} give
\begin{equation}
p(X|C_1C_2)= \mathcal{N}(\mu_2,\sigma^2_2),
\label{eq3.2.2.2}
\end{equation}
where
\begin{equation}
\mu_2=\frac{X_1 \left(1-\rho ^2\right)+D_{X_1} Y_2 \rho
   }{D_{X_1} \rho ^2-\rho ^2+1},
\label{eq3.2.2.3}
\end{equation}
and
\begin{equation}
\sigma^2_2=\frac{D_{X_1} \left(\left(\rho
   ^2-1\right)^2 - D_{X_1} \rho ^2 \left(-D_{Y_2}+\rho ^2-1\right)\right)}{\left(D_{X_1} \rho ^2-\rho ^2+1\right)^2}.
\label{eq3.2.2.4}
\end{equation}

$X_2$ can be decoded as: $X_2=\mu_2$, thus
\[\mathbb{E}d(X,X_2)=\sigma^2_2.\]

The distortion pair in the second stage is
\begin{equation}
(D_{X_2},D_{Y_2})=\left(\sigma^2_2,D_{Y_2}\right).
\label{eq3.2.2.5}
\end{equation}

\subsubsection{Third stage}
\label{sec3.2.3}
$(C_1,C_2)$ are known due to the broadcast advantage, thus according to (\ref{eq3.2.2.2}), $X-\mu_2\sim \mathcal{N}(0,\sigma^2_2)$. $X-\mu_2$ is encoded to $C_3$ using the random codebook argument, with R-D pair
\begin{equation}
R_3=\frac{1}{2}\log \frac{\sigma^2_2}{D_{X_3}},\quad D_{X_3}\leq \sigma^2_2.
\label{eq3.2.3.1}
\end{equation}

$X_3$ can be decoded as: $X_3=C_3+\mu_2$, thus
\begin{align*}
\mathbb{E}d(X,X_3)&=\mathbb{E}d(X-\mu_2,X_3-\mu_2) \\
&=\mathbb{E}d(X-\mu_2,C_3) \\
&=D_{X_3}.
\end{align*}

Similar calculations as Sec.~\ref{sec3.2.1} give
\begin{equation*}
p(Y|C_1C_2C_3)= \mathcal{N}(\mu_3,\sigma^2_3),
\label{eq3.2.3.2}
\end{equation*}
where $\mu_3$ and $\sigma^2_3$ are given in the Appendix.

$Y_3$ can be decoded as: $Y_3=\mu_3$, thus
\[\mathbb{E}d(Y,Y_3)=\sigma^2_3.\]

The distortion pair in the third stage is
\begin{equation}
(D_{X_3},D_{Y_3})=\left(D_{X_3},\sigma^2_3\right).
\label{eq3.2.3.3}
\end{equation}

\subsubsection{Comments on coding}
\label{sec3.2.4}
First, the whole coding procedure indicates that $\rho, D_{X_1}, D_{Y_2}, D_{X_3}$ must be known both to encoder and decoder before coding, since $(\mu_i,\sigma^2_i) \textrm{ for }i=1,2,3$ are required as pre-known parameters.

Second, special attention must be paid to the distortion constraints $D_{X_1}\leq 1, D_{Y_2}\leq \sigma^2_1, D_{X_3}\leq \sigma^2_2$, which is the assumptions we used for the calculation of the conditional probabilities. Due to the recursive structure of calculations, if any intermediate stage does not meet the distortion constraint, negative rate would occur (which obviously is not allowed) which breaks the recursive chain and leads to an incorrect result.

\subsection{Successive refinement of DiSAC2}
\label{sec3.3}
As we know, a Gaussian source with quadratic distortion is successively refinable on its \textit{\{rate, distortion\}} curve.  In this section, we show that a three-stage DiSAC2 with quadratic Gaussian case is also successively refinable on the \textit{\{sum-rate, distortion pair\}} surface which is characterized by the rate-distortion region of the DSC2~\cite{Wagner2008} (referred to as \textit{Wagner Surface} in this section).

\begin{theorem} For a three-stage DiSAC2 with Gaussian sources and quadratic distortion, $\left\{R_1,(D_{X_1},D_{Y_1})\right\}$, $\left\{R_1+R_2,(D_{X_2},D_{Y_2})\right\}$ and $\left\{R_1+R_2+R_3,(D_{X_3},D_{Y_3})\right\}$ all achieve the \{sum-rate, distortion pair\} surface which is characterized by the rate-distortion region of the DSC2, for any rate triplet $(R_1,R_2,R_3)$.
\label{theoremsuccessive}
\end{theorem}

\begin{proof}
From \cite{Wagner2008}, the minimum sum-rate for DSC2 is
\begin{equation}
R_{\textrm{DSC2}}(D_X,D_Y) = \frac{1}{2}\log \frac{\left(1-\rho ^2\right) \left(\sqrt{\frac{4 D_{X} D_{Y} \rho ^2}{\left(1-\rho ^2\right)^2}+1}+1\right)}{2 D_{X} D_{Y}}.
\label{eq3.3.1}
\end{equation}

At the first stage, the \{sum-rate, distortion pair\} is known from (\ref{eq3.2.1.1}) and (\ref{eq3.2.1.7}):
\begin{equation}
\left\{\begin{split}
& R_1 = \frac{1}{2}\log \frac{1}{D_{X_1}} \\
& (D_{X_1},D_{Y_1}) = \left(D_{X_1},\sigma^2_1\right).
\end{split}\right.
\label{eq3.3.2}
\end{equation}

Combining (\ref{eq3.3.2}) and (\ref{eq3.3.1}) leads to
\begin{equation*}
R_1=R_{\textrm{DSC2}}(D_{X_1},D_{Y_1}).
\label{eq3.3.3}
\end{equation*}

At the second stage, the \{sum-rate, distortion pair\} is known from (\ref{eq3.2.1.1}), (\ref{eq3.2.2.1}) and (\ref{eq3.2.2.5}):
\begin{equation}
\left\{\begin{split}
& R_1+R_2 = \frac{1}{2}\log \frac{\sigma^2_1}{D_{X_1}D_{Y_2}} \\
& (D_{X_2},D_{Y_2})=\left(\sigma^2_2,D_{Y_2}\right).
\end{split}\right.
\label{eq3.3.4}
\end{equation}

Combining (\ref{eq3.3.4}) and (\ref{eq3.3.1}) leads to
\begin{equation*}
R_1+R_2 = R_{\textrm{DSC2}}(D_{X_2},D_{Y_2}).
\label{eq3.3.5}
\end{equation*}

At the third stage, the \{sum-rate, distortion pair\} is known from (\ref{eq3.2.1.1}), (\ref{eq3.2.2.1}), (\ref{eq3.2.3.1}) and (\ref{eq3.2.3.3}):
\begin{equation}
\left\{\begin{split}
& R_1+R_2+R_3 = \frac{1}{2}\log \frac{\sigma^2_1\sigma^2_2}{D_{X_1}D_{Y_2}D_{X_3}} \\
& (D_{X_3},D_{Y_3})=\left(D_{X_3},\sigma^2_3\right).
\end{split}\right.
\label{eq3.3.6}
\end{equation}

Combining (\ref{eq3.3.6}) and (\ref{eq3.3.1}) leads to
\begin{equation*}
R_1+R_2+R_3 = R_{\textrm{DSC2}}(D_{X_3},D_{Y_3}).
\label{eq3.3.7}
\end{equation*}
\end{proof}

\begin{figure*}[!ht]
\centerline{\subfloat[]{\label{surface1}\includegraphics[width=3.0in]{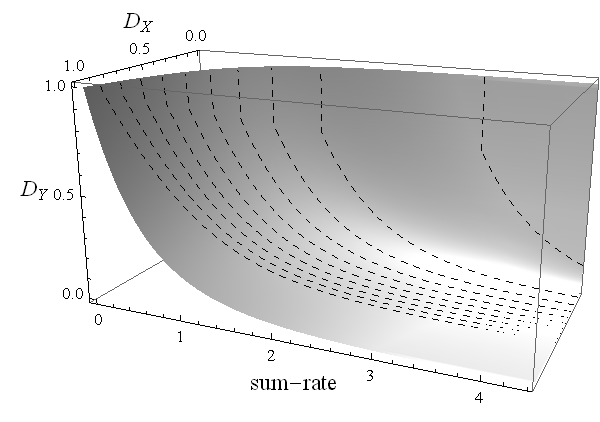}}
\\ \subfloat[]{\label{surface2}\includegraphics[width=3.0in]{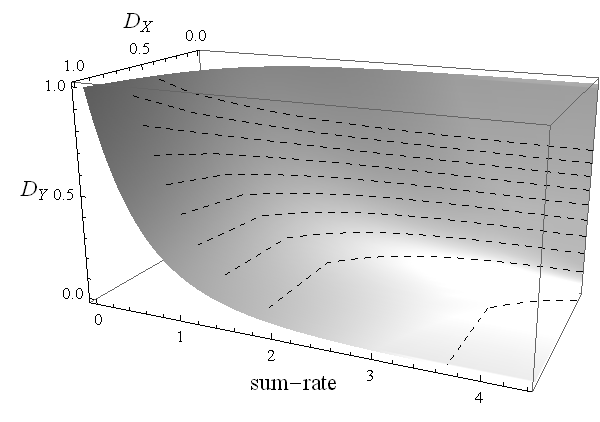}}}
\caption{Successive refinement of the DiSAC2: (a) Given a fixed $R_1$, the DiSAC2 with any additional rate $R_2$ operates along a one-dimensional curve on the \textit{Wagner Surface} (represented by the gray surface). The parameters chosen are $\rho=0.6$, and 10 different $R_1$. (b) Given a fixed $R_1+R_2$, the DiSAC2 with any additional rate $R_3$ operates along a one-dimensional curve on the \textit{Wagner Surface}. The parameters chosen are $R_1=0.5, \rho=0.6$, and 10 different $R_2$.}
\end{figure*}

Fig.~\ref{surface1} gives a visual illustration of Theorem~\ref{theoremsuccessive}: We choose an initial operating point $\left\{R_1,(D_{X_1},D_{Y_1})\right\}$ at the first stage of DiSAC2, and then send an additional rate $R_2$ at the second stage. As the dashed curves on the gray surface suggest, $\left\{R_1+R_2,(D_{X_2},D_{Y_2})\right\}$ with any $R_2$ operates along a one-dimensional curve on the \textit{Wagner Surface}. 10 different initial points chosen give 10 curves verifying the same idea: the DiSAC2 is successively refinable from the first stage to the second stage. A similar plot (Fig.~\ref{surface2}) shows that the DiSAC2 is also successively refinable from the second stage to the third stage.

\subsection{Additional theoretical properties}
\label{sec3.4}
\subsubsection{Feasible region of the distortion pair}
\label{sec3.4.1}
As we mentioned in Sec.~\ref{sec3.2.4}, the distortion constraints $D_{X_1}\leq 1, D_{Y_2}\leq \sigma^2_1, D_{X_3}\leq \sigma^2_2$ have to be satisfied. It is shown in this section that given a fixed $\rho$, not all distortion pairs are feasible. An extreme situation is when $\rho=1$ ($Y$ is just a copy of $X$), the feasible region of the distortion pair (abbreviated as \textit{feasible region} in the following) for the DiSAC2 is restricted to $D_{X_i}=D_{Y_i}$.

\begin{figure}[!hb]
\centering
\subfloat[Second stage, $\rho=0.3$.]{\label{feasible1}\includegraphics[width=1.7in]{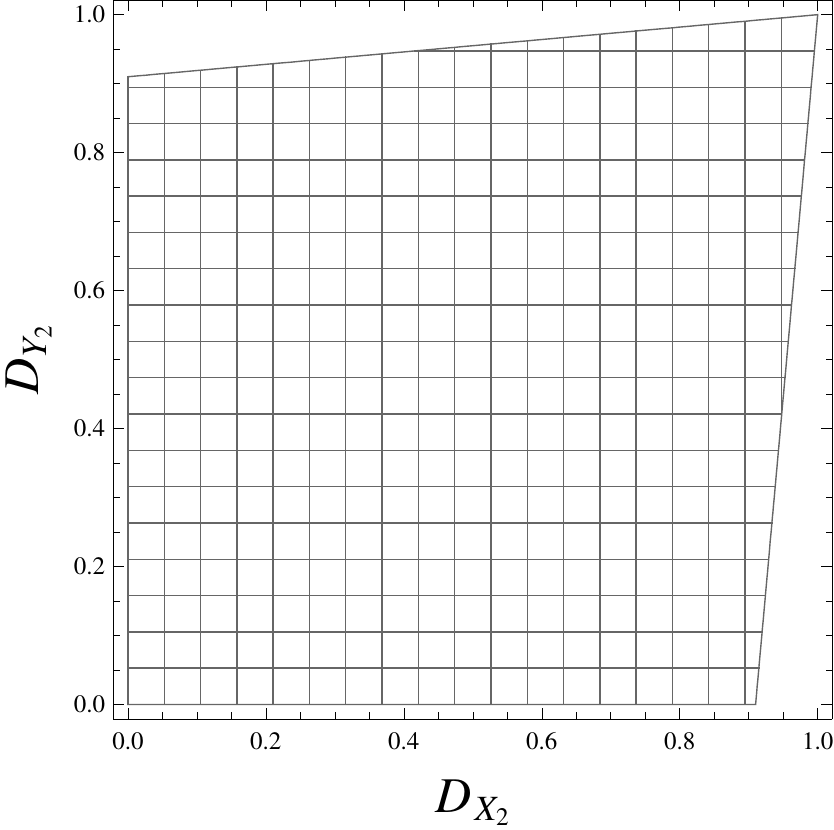}}
\hfil  \subfloat[Second stage, $\rho=0.7$.]{\label{feasible2}\includegraphics[width=1.7in]{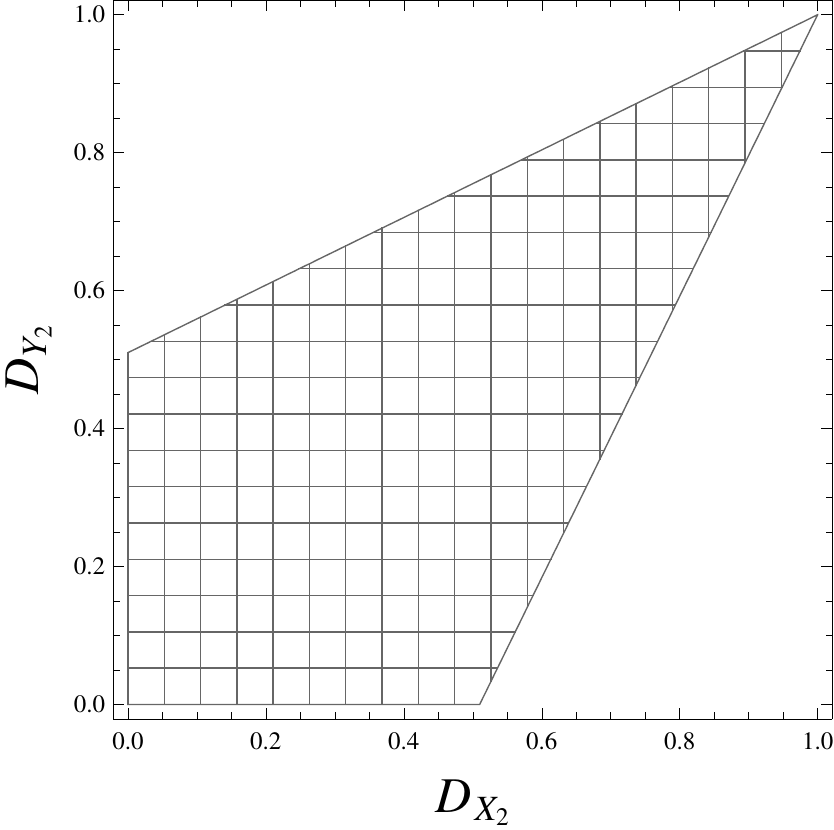}}
\\ \subfloat[Second stage, $\rho=0.9$.]{\label{feasible3}\includegraphics[width=1.7in]{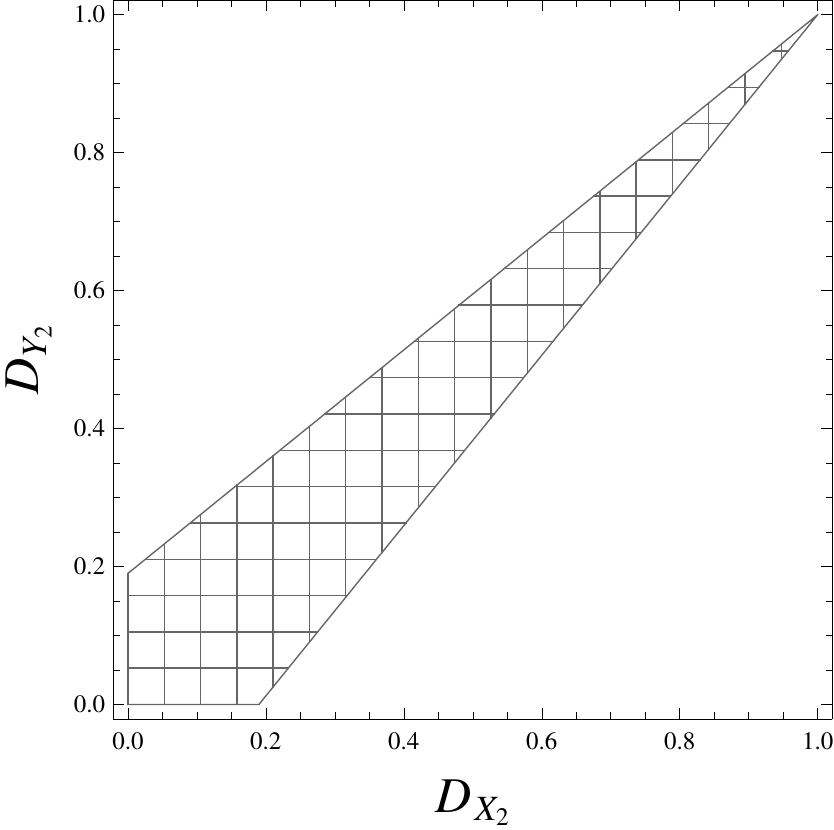}}
\hfil  \subfloat[Third stage, $\rho=0.7$.]{\label{feasible4}\includegraphics[width=1.7in]{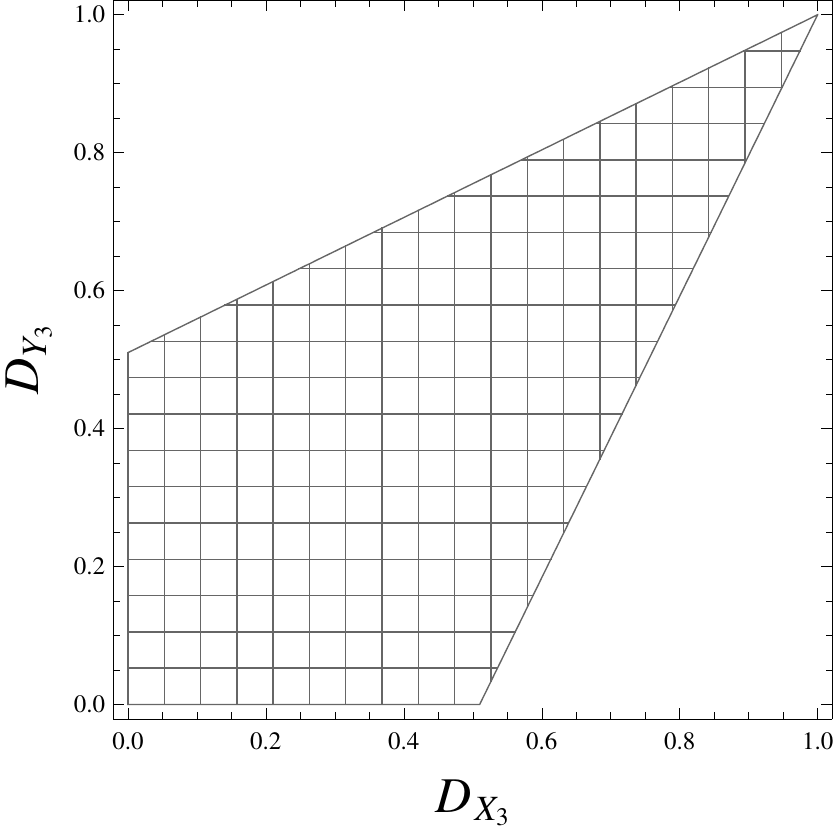}}
\caption{Feasible region of the distortion pair.}
\label{feasible}
\end{figure}

At the second stage of DiSAC2, we define the \textit{feasible region} as:
\begin{equation}
\left\{(D_{X_2},D_{Y_2}):\ D_{X_1}\leq 1\,\cap\,D_{Y_2}\leq \sigma_1^2\right\}.
\label{eq3.4.1.1}
\end{equation}

Similarly, the \textit{feasible region} at the third stage of DiSAC2 can be defined as:
\begin{equation}
\left\{(D_{X_3},D_{Y_3}):\ D_{X_1}\leq 1\,\cap\,D_{Y_2}\leq \sigma_1^2\,\cap\,D_{X_3}\leq \sigma_2^2\right\}.
\label{eq3.4.1.2}
\end{equation}

Fig.~\ref{feasible1} - Fig.~\ref{feasible3} show the \textit{feasible region} at the second stage drawn according to (\ref{eq3.4.1.1}). It can be seen that the \textit{feasible region} gradually shrinks to the axis $D_{X_2}=D_{Y_2}$ as $\rho$ increases, and finally ends up in the extreme situation when $\rho=1$. Particularly, $D_{X_2}=D_{Y_2}\in (0,1)$ is always feasible, regardless of $\rho$. Fig.~\ref{feasible4} illustrates the \textit{feasible region} at the third stage drawn according to (\ref{eq3.4.1.2}), which is shown to be identical with the one at the second stage (Fig.~\ref{feasible2}).

\subsubsection{Rate allocation between two channels}
\label{sec3.4.2}
The rate allocation between two channels is important as it is key for energy allocation. Recalling (\ref{eq3.2.1.1}), (\ref{eq3.2.2.1}), (\ref{eq3.2.3.1}) and Fig.~\ref{model}, we denote
\begin{equation*}
R_X=R_1+R_3\quad \textrm{and}\quad R_Y=R_2
\label{eq3.4.2.1}
\end{equation*}
as a representation of the rate for two channels in the three-stage DiSAC2 scheme.

The dashed line of Fig.~\ref{rateallocation} shows the admissible rate pair $(R_X,\,R_Y)$ of the three-stage DiSAC2, for the distortion pair $D_{X_3}=D_{Y_3}=0.5$ and $\rho=0.6$. It can be seen that the rate pair $(R_X,\,R_Y)$ is configurable while the sum-rate $R_X+R_Y$ is kept to be a constant, and there exists a particular point to achieve $R_X=R_Y$.

\begin{figure}[!ht]
\centering
\includegraphics[width=2.0in]{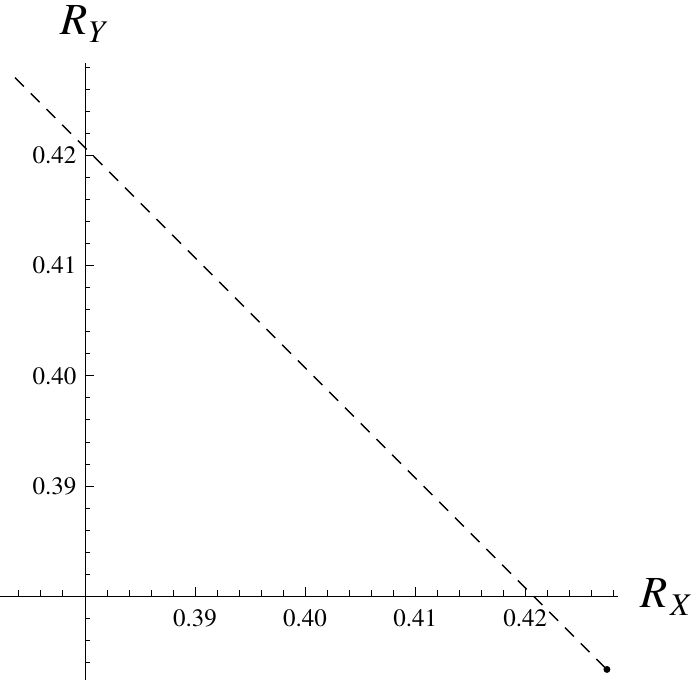}
\caption{The admissible rate pair $(R_X,\,R_Y)$ of the three-stage DiSAC2, for the distortion pair $D_{X_3}=D_{Y_3}=0.5$ and $\rho=0.6$.}
\label{rateallocation}
\end{figure}

\subsubsection{Potential energy saving}
\label{sec3.4.3}
let us revisit channel capacity as a function of energy. For the discrete time Gaussian channel, the capacity~\cite[P.249]{Cover2006} is
\begin{equation*}
C=\frac{1}{2}\log (1+\frac{E_s}{N_s}) \quad \textrm{bits per transmission},
\label{eq3.4.3.1}
\end{equation*}
where $E_s$ is the \textit{energy per sample}, and $N_s$ is the \textit{noise variance per sample}.

We consider further a band-limited channel (bandwidth $W$) with white noise (spectral density $N_0/2$), and each sample occupies a time interval $T$. In this case, the \textit{noise variance per sample} is $\frac{N_0}{2}2W\frac{T}{2WT} = N_0/2$. Hence, the capacity is
\begin{equation*}
C = \frac{1}{2}\log (1+\frac{E_s}{N_0/2})\quad \textrm{bits per sample},
\label{eq3.4.3.2}
\end{equation*}
which can be rewritten as
\begin{equation*}
E_s = \left(e^{2C}-1\right) N_0/2 \quad \textrm{per sample}.
\label{eq3.4.3.3}
\end{equation*}

Due to the successive approximation structure of the DiSAC2, there is no interference between successive transmissions. Therefore, assuming that the channels from each source to the base station are independent, Gaussian band-limited, with the same bandwidth $W$ and noise spectral density $N_0/2$, the sum-energy for the three-stage DiSAC2 can be modeled as
\begin{equation}
\begin{split}
E_{\textrm{DiSAC2}} &= \sum_{i=1}^3 \left(e^{2C_i}-1\right) N_0/2 \\
 &= \sum_{i=1}^3 \left(e^{2R_i}-1\right) N_0/2 \\
 &= \left(\frac{1}{D_{X_1}}+\frac{\sigma_1^2}{D_{Y_2}}+\frac{\sigma_2^2}{D_{X_3}}-3\right) N_0/2,
\end{split}
\label{eq3.4.3.4}
\end{equation}
where the second equality follows assuming that all rates match the capacity for each channel usage (the Shannon channel coding theorem is met perfectly).

For the DSC2 scheme, with similar assumptions as above, we can get
\begin{equation}
\begin{split}
E_{\textrm{DSC2}} &= \left(e^{2R_{X}}+e^{2R_{Y}}-2\right) N_0/2 \\
 &\geq \left(\sqrt{\frac{\left(1-\rho ^2\right) \left(\sqrt{\frac{4 D_{X} D_{Y} \rho ^2}{\left(1-\rho ^2\right)^2}+1}+1\right)}{2 D_{X} D_{Y}}}-1\right) N_0 \\
 & \qquad\qquad\qquad\qquad\qquad\qquad\qquad,
\end{split}
\label{eq3.4.3.5}
\end{equation}
where the second inequality follows from the fact that $R_{X}+R_{Y}$ is constrained by the minimum sum-rate (\ref{eq3.3.1}).

\begin{figure}[!ht]
\centering
\includegraphics[width=2.8in]{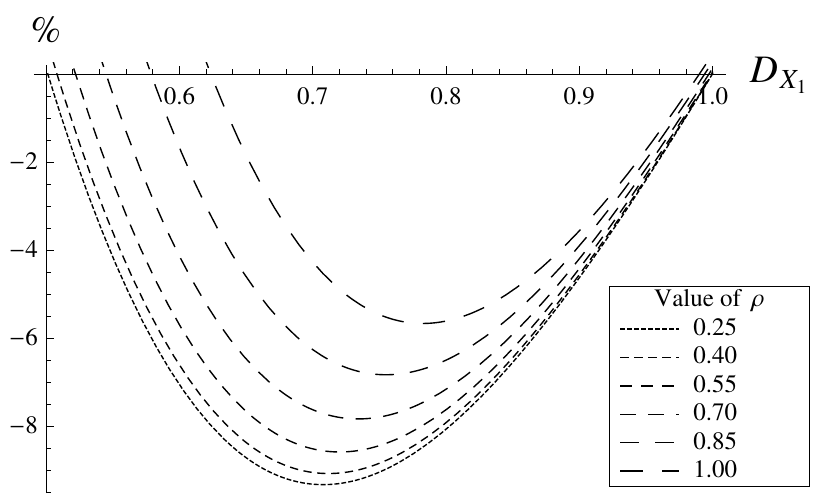}
\caption{The energy saving $(E_{{\rm DiSAC2}}-\min E_{{\rm DSC2}})/E_{{\rm DiSAC2}}$ of the three-stage DiSAC2, for $D_{X}=D_{Y}=0.5$ and different $\rho$. The DiSAC2 is better than DSC2 in terms of transmission power. Notice that $D_{X_1}$ has a lower bound to satisfy the distortion constraints, which varies with $\rho$.}
\label{saving}
\end{figure}

To compare the two schemes, we assume that the goal for both is to convey the source $X$ and $Y$ to the receiver while attaining a given distortion pair $(D_X,D_Y)$. Thus, by plugging $D_{X_3}=D_X,D_{Y_3}=D_Y$ into (\ref{eq3.4.3.4}), $E_{\textrm{DiSAC2}}$ can be reduced to a function of a single variable $D_{X_1}$. We plot $\left(E_{\textrm{DiSAC2}}-\min E_{\textrm{DSC2}}\right)/E_{\textrm{DiSAC2}}$ for $D_{X}=D_{Y}=0.5$ and different $\rho$ in Fig.~\ref{saving}, in which the transmission power of the three-stage DiSAC2 is less than DSC2.

Since we use the non-interference model for the DSC2 scheme, which essentially consumes less than any other interference model (e.g., Gaussian multiple access channels~\cite[P.405]{Cover2006}), Fig.~\ref{saving} suggests that the three-stage DiSAC2 is potentially better than any conventional DSC2 scheme, in terms of transmission power. Further study of this energy saving effect is still in progress.

\section{Extension to the case of many encoders}
\label{sec4}
We have addressed the DiSAC2 scheme for the two-encoder case in the previous sections. By similar reasoning, we can extend this to a general scheme with any number of encoders (\textit{DiSACn}). Fig.~\ref{threeENC} illustrates the setup for three encoders --- DiSAC3: three encoders take turns to send messages, the BS successively decodes the messages to get better reconstruction of $(X,Y,Z)$ stage by stage. The central principle to design the coding procedure for such a scheme is: each encoder should fully exploit the joint information between the local source, and the coded messages that have already been sent (including the overheard messages and self-generated messages).

\begin{figure}[!ht]
\centering
\includegraphics[width=3.2in]{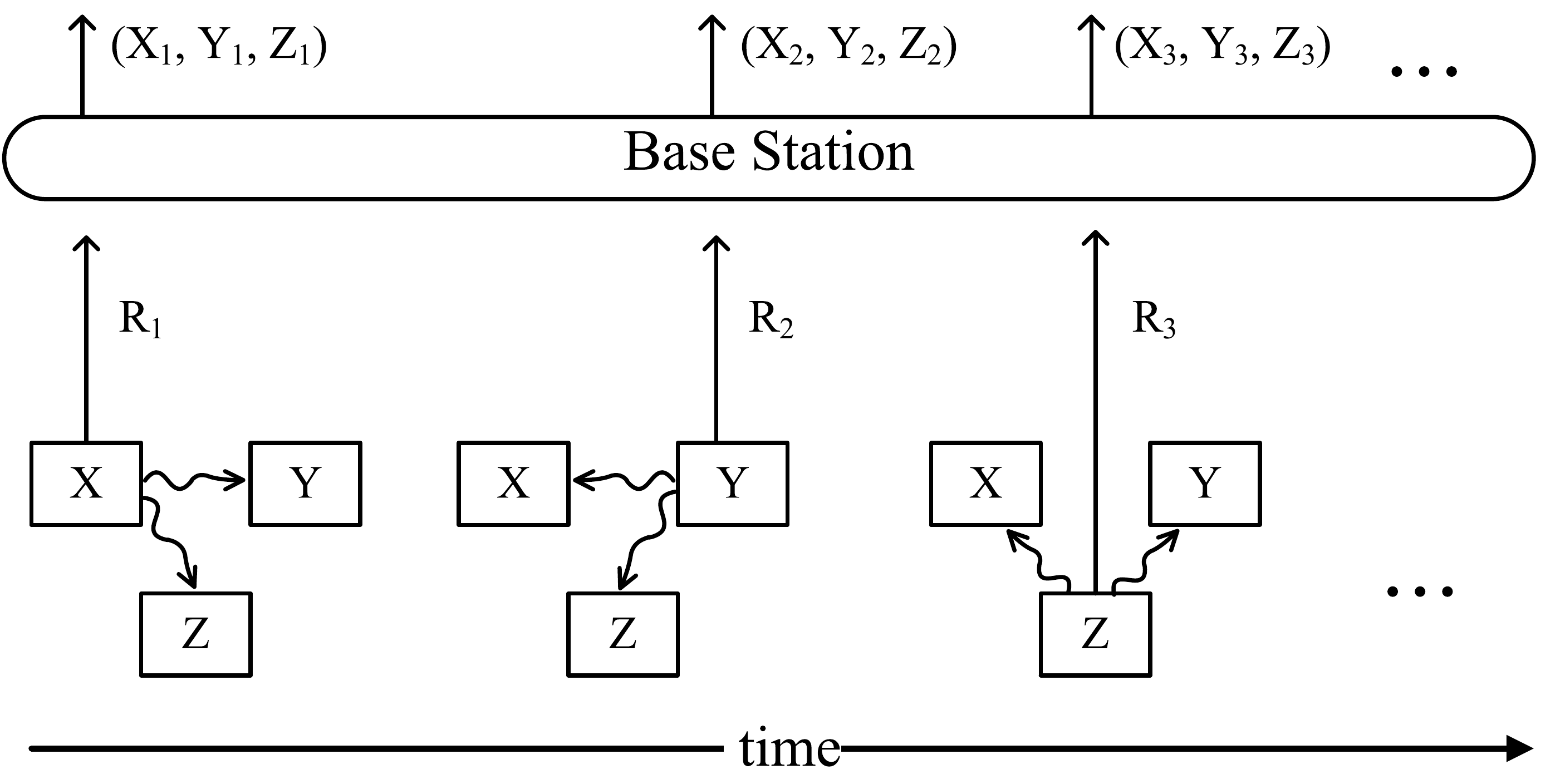}
\caption{Three-encoder distributed successive approximation coding using broadcast advantage. The straight lines represent the transmission from the encoder to the base station, and the wavy lines represent the same message that is overheard by the other two encoders.}
\label{threeENC}
\end{figure}

The minimum sum-rate for distributed source coding with three and more encoders is determined in \cite{Wagner2008} for a symmetric case: the source components are Gaussian, exchangeable, positively correlated, and all the target quadratic distortions are equal. It turns out that the scheme in Fig.~\ref{threeENC} with three stages can achieve this minimum sum-rate\footnote{Proof is omitted due to space, see \cite{Chen2010b} for details.}. It is interesting to study the successive refinement of this setup for this symmetric case and also more general conditions, which is the subject of on-going work~\cite{Chen2010b}.

\section{Conclusions}
\label{sec5}
We introduced a new distributed source coding scheme for two encoders (DiSAC2), which takes advantage of the free broadcast nature present in wireless networks. The coding procedure imitates the ping-pong game, and has a successive approximation structure. For the quadratic Gaussian case, it is proved that the three-stage DiSAC2 is successively refinable on the \textit{\{sum-rate, distortion pair\}} surface, which is characterized by the rate-distortion region of DSC2. The transmission power analysis also shows a possible energy saving over a conventional distributed coding setup.

We are working on the successive refinement conjecture of the $n$-stage DiSAC2, which is a generalization of Theorem~\ref{theoremsuccessive} to the case $n$. The extension of DiSAC2 to three and more encoders is illustrated, and studying its properties is left for future work. Last but not least, the broadcast advantage is expected to simplify the encoding/decoding in image coding, therefore indicating that such a scheme has potential in multi-camera networks. A practical distributed image coding method will be the best verification of its theoretical value.

\section*{Acknowledgement}
The authors would like to thank Prof. Michael Gastpar for the insightful discussion on the initial draft of this paper.

\appendix[The explicit expression of $\mu_3$ and $\sigma^2_3$]
{\footnotesize
\begin{multline*}
\mu_3=\Big( Y_2 \left(\rho ^2-1\right) \left(\rho ^2 D_{X_1}-\rho^2+1\right)- \\
    \rho D_{Y_2} \left(X_3 \left(\rho ^2 D_{X_1}-\rho^2+1\right)+
   X_1 \left(\rho ^2-1\right)\right)\Big) \bigg/ \\
   \Big(\rho ^2 D_{X_1}\left(-D_{Y_2}+\rho ^2-1\right)-\left(\rho ^2-1\right)^2\Big).
\end{multline*}}
{\footnotesize
\begin{multline*}
\sigma^2_3= D_{Y_2} \left(\rho ^2 D_{X_1}-\rho ^2+1\right) \Bigg(
    \left(1-\rho ^2\right) \left(\rho ^2 D_{X_3} D_{Y_2}+\left(\rho^2-1\right)^2\right)+ \\
   \rho ^2 D_{X_1} \left(D_{Y_2} \left(\rho ^2 D_{X_3}-\rho
   ^2+1\right)+\left(\rho ^2-1\right)^2\right)\Bigg)\Bigg/  \\
    \left(\rho ^2 D_{X_1} \left(D_{Y_2}-\rho
   ^2+1\right)+\left(\rho ^2-1\right)^2\right)^2.
\end{multline*}}

\bibliographystyle{IEEEtran}
\bibliography{IEEEabrv,F:/Work/LCAV/Papers/papers}

\end{document}